\documentclass{article}

\usepackage{arxiv}
\usepackage{xcolor}
\usepackage{amsmath}
\usepackage{amsthm}

\definecolor{AUXLblue}{RGB}{ 51,145,202}
\definecolor{AUXLlightblue}{RGB}{144,196,231} 

\definecolor{AUXLorange}{RGB}{246,174, 60}
\definecolor{AUXLlightorange}{RGB}{254,216,177}

\definecolor{AUXLpurple}{RGB}{135,119,175}
\definecolor{AUXLlightpurple}{RGB}{226,178,255}

\definecolor{AUXLyellow}{RGB}{255,221,  0}
\definecolor{AUXLlightyellow}{RGB}{255,238,127}

\definecolor{AUXLgreen}{RGB}{120,180,88}            
\definecolor{AUXLlightgreen}{RGB}{202,240,142} 

\usepackage{hyperref}
\usepackage[
			capitalize,
			nameinlink,
			noabbrev,
		   ]{cleveref}

\newtheorem{lemma}{Lemma}

\newtheorem{definition}[lemma]{Definition}

\newtheorem{observation}{Observation}

\newtheorem{theorem}{Theorem}

\usepackage{authblk}
\usepackage[utf8]{inputenc} 
\usepackage[T1]{fontenc}    
\usepackage{url}            
\usepackage{booktabs}       
\usepackage{amsfonts}       
\usepackage{nicefrac}       
\usepackage{microtype}      
\usepackage{lipsum}		
\usepackage{graphicx}
\usepackage{doi}

\usepackage{booktabs}
\usepackage{algorithm}
\usepackage{algorithmic}
\newcommand{\cI}{{\cal I}}

\newcommand{\cM}{{\cal M}}

\newcommand{\R}{\mathbb{R}}

\newcommand{\barT}{\bar{T}}
\newcommand{\bart}{\bar{t}}
\newcommand{\barx}{\bar{x}}
\newcommand{\barX}{\bar{X}}

\newcommand{\E}{\mathrm{E}}

\newcommand{\ms}{\mathrm{MS}}

\newcommand{\opt}{\mathrm{OPT}}
\newcommand{\fr}{\mathrm{FR}}
\newcommand{\pay}{\mathrm{PAY}}

\newcommand{\optv}{v^*}

\newcommand{\optc}{c^*}

\newcommand{\optt}{t^*}
\newcommand{\optmint}{t^{*(min)}}
\newcommand{\optT}{T^*}
\newcommand{\optX}{X^*}

\newcommand{\predL}{\hat{L}}

\newcommand{\predt}{\hat{t}}
\newcommand{\predX}{\hat{X}}
\newcommand{\predx}{\hat{x}}
\newcommand{\predmint}{\hat{t}^{(min)}}
\newcommand{\predT}{\hat{T}}
\newcommand{\predV}{\hat{V}}
\newcommand{\predv}{\hat{v}}
\newcommand{\predi}{\hat{i}}

\newcommand{\predl}{\hat{l}}
\newcommand{\predC}{\hat{C}}
\newcommand{\predc}{\hat{c}}

\newcommand{\sqmechanism}{^{\sqrt{\ }}}

\title{Mechanism Design with Predictions\thanks{This work was supported by Science and Technology Innovation 2030 – ``The Next Generation of Artificial Intelligence" Major Project No.2018AAA0100900.} }

\renewcommand{\shorttitle}{Mechanism Design with Predictions}

\author[1]{Chenyang Xu}
\author[2, 3]{Pinyan Lu}

\affil[1]{\footnotesize College of Computer Science, Zhejiang University}
\affil[2]{\footnotesize ITCS, Shanghai University of Finance and Economics}
\affil[3]{\footnotesize Huawei TCS Lab}
\affil[ ]{\texttt{xcy1995@zju.edu.cn, lu.pinyan@mail.shufe.edu.cn}}

\date{}

\renewcommand{\headeright}{}
\renewcommand{\undertitle}{}

\begin{document}
\maketitle

\begin{abstract}
	Improving algorithms via predictions is a very active research topic in recent years. This paper initiates the systematic study of mechanism design in this model. In a number of well-studied mechanism design settings, we make use of imperfect predictions to design mechanisms that perform much better than traditional mechanisms if the predictions are accurate (consistency), while always retaining worst-case guarantees even with very imprecise predictions (robustness). Furthermore, we refer to the largest prediction error sufficient to give a good performance as the error tolerance of a mechanism, and observe that an intrinsic tradeoff among consistency, robustness and error tolerance is common for mechanism design with predictions.
\end{abstract}


\section{Introduction}

A recently popular trend in algorithmic design is augmenting (online) algorithms with imperfect predictions. This line of work suggests that there exists an opportunity to bypass the worst-case lower bounds of online problems, which are caused by the uncertainty of the future. In this setting, the algorithm is given access to error-prone predictions, and its performance is bounded in terms of the quality of the predictions.
The algorithm should perform better than the worst-case bound with accurate predictions, and never perform much worse than the best pure online algorithm even if the prediction error is large.
Many classic online problems have been considered in the context, such as ski rental~\cite{DBLP:conf/nips/PurohitSK18,DBLP:conf/icml/GollapudiP19,DBLP:conf/icml/AnandGP20},
caching~\cite{DBLP:conf/icml/LykourisV18,DBLP:conf/soda/Rohatgi20,DBLP:conf/icalp/JiangP020}, 
and scheduling~\cite{DBLP:conf/nips/PurohitSK18,DBLP:conf/soda/LattanziLMV20,DBLP:conf/spaa/Im0QP21,DBLP:conf/icml/0001X21}.


Mechanism design and online algorithm design share some similarities. Both of them need to deal with missing information. Due to the uncertainty about the future or the agents' private information, algorithms (mechanisms) have to be overly cautious and thereby worst-case bounds arise. Thus, it is interesting to investigate if predictions can help with mechanism design.
This paper proposes a study of mechanism design with predictions. For generality and simplicity, we directly use agents' private information as predictions. This is justifiable in many applications. For example, in repeated auctions, we may use historical bidding records as the predictions. We note that there are several related works~\cite{DBLP:conf/nips/MedinaV17,DBLP:conf/nips/AntoniadisGKK20}. In~\cite{DBLP:conf/nips/MedinaV17}, repeated posted price auctions were considered. 
In~\cite{DBLP:conf/nips/AntoniadisGKK20}, the authors mainly focused on developing an online bipartite matching algorithm with predictions and noticed that the algorithm can be converted to a truthful mechanism.
In this paper, we consider mechanism design with predictions more systematically and investigate a number of different mechanism design settings.

\subsection{Challenge}

We follow the terminology in~\cite{DBLP:conf/nips/PurohitSK18} which is now standard: say a mechanism's
\emph{consistency} and \emph{robustness} are its approximation ratios for accurate predictions and for arbitrarily inaccurate predictions respectively.
Due to the truthfulness requirement in mechanism design, it is very subtle to use predictions to get a good consistency and robustness.



A mechanism consists of two algorithms: the allocation algorithm and the payment algorithm. Due to truthfulness, the allocation algorithm needs to satisfy a certain monotone property, which is a global property of the algorithm. It is usually infeasible to change the allocation for certain inputs based on the predictions without changing the allocation for other inputs. For example, a widely-used approach to ensure robustness in online algorithm design is switching to pure online algorithms when the predictions are found to be unreliable. However, in mechanism design, switching to traditional mechanisms when the prediction error is large may hurt truthfulness. The agent who benefits more from the traditional mechanism may misreport the private information such that the prediction error looks large.

Thus, we need to design the allocation algorithm with predictions as a whole mapping to satisfy the monotone property. This gives a big challenge to maintain a good consistency and robustness.
In other words, it is much less flexible to design a truthful mechanism than an (online) algorithm.

\subsection{Our Contributions}
We study four very different and well-studied mechanism design problems and observe that predictions are indeed helpful.



\paragraph{Revenue-Maximizing Single-Item Auction.} Maximizing the revenue is one of the most fundamental problem in auction design. There is a rich literature (e.g.~\cite{DBLP:journals/mor/Myerson81,DBLP:conf/soda/GoldbergHW01,DBLP:conf/soda/AzarDMW13}). We consider the revenue-maximizing single-item auction and compare the revenue to the highest bid, the most ambitious benchmark. It is well-known that there is no good mechanism in worst-case mechanism design with respect to this goal.
If we assume that all bids belong to $[1,h]$, no deterministic truthful mechanism has an approximation ratio better than $h$ (~\cite{DBLP:conf/soda/GoldbergHW01}). 

With perfect predictions, it is trivial to achieve $1$-approximation since we can simply run anonymous pricing scheme and set the price as the highest predicted value. However, this mechanism is very fragile. The approximation ratio drops to infinity even if there is only a tiny error in the predictions. This is of course highly undesirable. To address the issue, we propose a notion of \emph{error tolerance} to measure how much prediction error the mechanism can tolerate to get a reasonable good approximation. Let $\eta\ge 1$ be the relative prediction error, where $\eta=1$ means that there is no error. We give a deterministic truthful mechanism with $\gamma$-consistent and $h$-robust,  where the robustness ratio matches the worst-case bound of $h$ in the traditional setting and the consistency ratio $\gamma\ge 1$ is a parameter we can choose. The approximation ratio smoothly increases as a function of $\gamma \eta$ when $\eta\leq \gamma$, and then has a big drop after $\eta> \gamma$. Therefore, there is a tradeoff between the consistency and the error tolerance. Moreover, we prove that such a tradeoff is necessary and in some sense optimal for all deterministic truthful mechanisms.
Our mechanism is simple and practical. It is the second price auction with individual reserve prices, where these reserves are set based on the predictions.

\paragraph{Frugal Path Auction.}
Path auction is a reverse auction, where the auctioneer needs to buy a path and pays the edges in the path. The goal here is to minimize the total payment and the benchmark is the second cheapest path. This is a classic problem in frugal mechanism design. The problem was coined by~\cite{DBLP:conf/soda/ArcherT02}. They showed that the VCG mechanism obtains an approximation ratio of $\Theta(n)$, where $n$ is the number of agents, and the ratio is the best possible~\cite{DBLP:conf/soda/ElkindSS04}. We obtain a deterministic truthful mechanism with $2\gamma$-consistent and $(n^2/\gamma)$-robust. In terms of error tolerance, the approximation ratio smoothly increases as a function of $\gamma (1+\eta)$ as long as $\eta \leq \gamma$.
Here we observe a three way tradeoff among consistency, robustness and error tolerance.
The mechanism is the generalized VCG mechanism (a.k.a affine maximizer), where the weights for different agents are set based on the predictions.

Besides the VCG mechanism, \cite{DBLP:conf/focs/KarlinKT05} proposed $\sqmechanism$-mechanism for frugal path auction.
The approximation ratio of $\sqmechanism$-mechanism is still $\Theta(n)$, but it can outperform the VCG mechanism in some graphs. Later, this technique was generalized to more problems in frugal mechanism design~\cite{DBLP:conf/focs/ChenEGP10}. We can also apply our technique to $\sqmechanism$-mechanism to get a similar improvement when predictions are given.

\paragraph{Truthful Job Scheduling.}
Truthful mechanism for scheduling unrelated machines is the center problem of the very first algorithmic game theory paper~\cite{DBLP:journals/geb/NisanR01}, whose goal is to minimize the makespan. This problem is very different from the previous two settings in two aspects: it is a multidimensional mechanism design problem and the objective is not related to the payment. In~\cite{DBLP:journals/geb/NisanR01}, the authors showed that the VCG mechanism gives an approximation ratio of $m$, where $m$ is the number of machines, and proved a lower bound of $2$ for any deterministic truthful mechanism. They conjectured that no deterministic truthful mechanism has an approximation ratio better than $m$. Many papers worked on closing the gap~\cite{DBLP:conf/soda/ChristodoulouKV07,DBLP:conf/mfcs/KoutsoupiasV07,DBLP:conf/sigecom/AshlagiDL09,DBLP:conf/sagt/GiannakopoulosH20,DBLP:journals/corr/abs-2007-04362,DBLP:conf/stoc/0001KK20}.
The closest result so far proves a lower bound of $\Omega(\sqrt{m})$~\cite{DBLP:conf/stoc/0001KK20}.

For this problem, we give a deterministic truthful mechanism with approximation ratio of $O(\min\{\gamma \eta^2, \frac{m^3}{\gamma^2}\})$,
where again $\gamma$ is the consistency parameter we can choose and $\eta$ is the prediction error. Here, we compute an (approximate) optimal allocation based on the predicted information and use that allocation as a guide for the mechanism.

\paragraph{Two-Facility Game on a Line.}Finally, we consider a mechanism without money: two-facility game on a line. This setting was coined by~\cite{DBLP:conf/sigecom/ProcacciaT09}, where the authors gave an upper bound of $n-2$ and a lower bound of 1.5 for deterministic truthful mechanisms.
The lower bound was later improved to 2~\cite{DBLP:conf/wine/LuWZ09} and $(n-1)/2$~\cite{DBLP:conf/sigecom/LuSWZ10}. Finally, \cite{DBLP:journals/teco/FotakisT14} showed a tight lower bound of $n-2$.

Since the space of truthful mechanism without money is more restricted, it is even more difficult to make use of predictions here.
We get a deterministic truthful mechanism with $(1+n/2)$-consistent and $(2n-1)$-robust, whose consistency ratio is slightly better than the best known mechanism. Whether there is a mechanism with $o(n)$-consistent and a bounded robustness is a very interesting open question.

\section{Preliminaries}\label{sec:pre}

This section introduces the terminology necessary to understand the paper. An expert can skip it directly. We take the single-item sealed-bid auction for an example. In the auction, there is a seller that has a single good and several bidders who are interested in buying the good. Each bidder $i$ has a private value $v_i$ representing the maximum willingness-to-pay for the item. Each bidder $i$ privately tells the auctioneer a bid $b_i$, while the auctioneer decides who is the winner (i.e. \emph{the allocation rule}) and how much he needs to pay (i.e. \emph{the payment rule}). Note that the bidder could misreport his maximum willingness-to-pay, namely, it is possible that $b_i\neq v_i$. Say the bidder who sets $b_i=v_i$ is a truthtelling bidder. 

The utility of a bidder is defined as follows. If he is the winner, the utility is $v_i-p_i$, where $p_i$ is the price he needs to pay. Otherwise, the utility is $0$. 

\begin{definition}(\cite{DBLP:journals/eatcs/RoughgardenI17})\label{def:DSIC}
	A mechanism is \emph{truthful} if for any bidder $i$, setting $b_i=v_i$ always maximizes his utility regardless of other bidders' bids, and the utility of any truthtelling bidder is non-negative.
\end{definition}

Now we state a widely-used theorem proposed by~\cite{DBLP:journals/mor/Myerson81}, which helps to design a truthful mechanism or prove the truthfulness of a mechanism for single-parameter environments (e.g. single-item auctions). 

\begin{definition} (Monotonicity)\label{def:Monotonicity}
	An allocation rule is monotone if the winner still wins when he increases the bid unilaterally. 
\end{definition}

\begin{theorem} (Myerson's Lemma~\cite{DBLP:journals/mor/Myerson81})\label{thm:myerson}
	For a single-parameter environment, the mechanism is truthful only if its allocation rule is monotone, while for any monotone allocation rule, there exists an unique payment rule which makes the mechanism truthful. Moreover, such a payment rule can be given by an explicit formula. 
\end{theorem}

For single-item auctions, the unique payment is the winner's threshold bid.

\begin{definition} (Threshold Bid)\label{def:threshold_bid}
	Given a single-item auction mechanism and the bid of each bidder, the threshold bid of a bidder is the minimum bid that he could make and win the auction when all other bidders fix their bids.
\end{definition}


\section{Revenue-Maximizing Single-Item Auction}\label{sec:auction}

In this section, we consider revenue-maximizing single-item auction. 
There is one item and $n>1$ bidders. Each bidder $i$ has a private value $\optv_{i} \in [1,h]$ and reports a bid $b_i\in[1,h]$. Use a $n$-dimensional vector $X$ to denote the allocation of the item, where each entry $x_i$ is either $1$ or $0$ and it equals $1$ only if bidder $i$ wins.  
Each bidder $i$ has an utility $u_i=x_i \cdot \optv_i - p_i$, where $p_i$ is the payment. The auctioneer aims to design a truthful mechanism that maximizes the selling price, and the benchmark is the highest private value.
In mechanism design with predictions, we are given access to the predictions of bidders' private values, denoted by $\predV = \{\predv_i\}_{i\in [n]}$. The predictions are erroneous and we define a natural prediction error $\eta:=\max_{i\in [n]}\{\frac{\optv_{i}}{\predv_{i}}, \frac{\predv_{i}}{\optv_{i}}\}$.

\begin{theorem}\label{thm:auction_main}
	There exists a deterministic truthful mechanism parameterized by $\gamma \geq 1$ with approximation ratio at most $\min \{ f(\eta),h \}$ where 
	$$
	f(\eta):= \left\{
	\begin{aligned}
	&\gamma\eta & \eta \leq \gamma \;\\
	&\max\{ \gamma^2\eta^2, \frac{h\eta}{\gamma^2} \}  & \eta > \gamma. \\
	\end{aligned}
	\right.
	$$
\end{theorem}

From Theorem~\ref{thm:auction_main}, the claimed mechanism has a robustness ratio $h$, which is the best possible in the traditional setting. The consistency ratio of the mechanism is determined by a parameter $\gamma$. Notice that the function $f(\eta)$ is piecewise and has a ``jump'' when $\eta=\gamma$.
We refer to the length of the first piece in $f(\eta)$ as the error tolerance.
Later in this section, we will prove that this jump point is unavoidable for any deterministic truthful mechanism. 
In order to communicate our main ideas more clearly, this section simplifies the algorithm and the analysis by assuming that the number of bidders is at least $3$. In~\Cref{sec:auction_n2}, we show that removing this assumption still gets the claimed approximation ratio.

\subsection{Mechanism}

We start by giving some intuitions. 
For simplicity of notation, reindex the bidders in the decreasing order of their predicted values. So now bidder $1$ has the largest predicted value $\predv_1$.
Notice that if the predictions are error-free, letting bidder~$1$ win and charging him $\predv_1$ gives the optimal solution. Inspired by this, a natural idea is the following. If the reported bid $b_1$ is no less than the predicted value $\predv_1$, let bidder $1$ win. Otherwise, ignore bidder $1$ and let the remaining bidder with the highest bid win.

Clearly, this allocation rule is monotone. Then due to~\Cref{thm:myerson}, we can ensure the truthfulness by charging the winner the threshold bid. More precisely, if bidder $1$ wins, the payment is $\predv_1$, while if bidder $i$ $(\neq 1)$ wins, the payment is the second highest remaining bid. 
The mechanism has a consistency ratio of $1$, meaning that it gets the optimal value $\predv_1$ when $\eta=1$. However, the approximation ratio will increase to $h$ immediately even if $\eta$ is only slightly larger than $1$. Thus, although the robustness ratio of the mechanism is $h$ which is the best possible in the traditional setting, we cannot refer to it as an ideal mechanism due to the tiny error tolerance. 

One way to get a mechanism with a larger error tolerance is setting a lower bar for bidder $1$. For example, let bidder $1$ win if $b_1$ is at least $\predv_1/ \gamma$, where $\gamma$ is a parameter $\geq 1$. The mechanism will have a comparable performance when $\eta \leq \gamma$, because bidder $1$ will always win and the threshold bid is at least $\predv_1/ \gamma$.  We build on this idea to give our mechanism. Moreover, we refine the way we handle other bidders such that the approximation ratio will not increase to $h$ immediately once the prediction error becomes larger than $\gamma$. The mechanism is described in~\Cref{alg:auction}.

\begin{algorithm}[tb]
	\caption{\; Single-Item Auction with Predictions}
	\label{alg:auction}
	\textbf{Input}: The predicted private values $\predV=\{\predv_i\}_{i\in [n]}$, the bids $B=\{b_i\}_{i\in [n]}$ and a parameter $\gamma \geq 1$. \\
	\textbf{Output}: The winner and the payment.  
	  
	\begin{algorithmic}[1] 
		\STATE Reindex the bidders in the decreasing order of their predicted values.
		\STATE Assign a bar $br(i)$ to each bidder $i$. Initially, $br(i) \leftarrow 1$ $\forall i \in [n]$.  
		\IF { $\forall i \in [n], \predv_{i} \geq \predv_1 / \gamma^2$ }
		\STATE Update the bar of bidder 1: $br(1) \leftarrow \max\{\predv_1/\gamma,1 \}$.
		\ELSE
		\STATE Update the bar of bidder 1: $br(1) \leftarrow \max\{\predv_1/\gamma,1 \}$.
		\STATE For each bidder $i\neq 1$ with $\predv_{i} \geq \predv_1 / \gamma^2$, update the bar: $br(i) \leftarrow \max\{\predv_1/\gamma^2,1\}$.
		\ENDIF
		\STATE Define a bidder set $S:= \{ i \in [n] | b_i \geq br(i) \}$.
		\STATE \textbf{return} the bidder $j\in S$ with the highest bid and the threshold bid $\theta(j)$.
	\end{algorithmic}
\end{algorithm}

\subsection{Analysis}\label{sec:auction_analysis}

Observing that in~\Cref{alg:auction}, increasing any bidder's bid will only improve the chance that he wins, the allocation rule is monotone. Hence, we have the following lemma.
\begin{lemma}\label{lem:auction_DSIC}
	\Cref{alg:auction} is a truthful mechanism.
\end{lemma}
Now we focus on the approximation ratio analysis. 
We first give a simple proof of the robustness ratio $h$, and then show that the approximation ratio is at most $f(\eta)$.

\begin{lemma}\label{lem:auction_h}
	\Cref{alg:auction} has an approximation ratio at most~$h$.
\end{lemma}

\begin{proof}
	The upshot is that for any bidder $i$, $\theta(i)\geq br(i)$, because if $b_i < br(i)$, bidder $i$ will not be included in $S$ and lose the auction. 
	Since $br(i) \geq 1$ for any bidder $i$, the winner's payment is at least $1$.
	Note that $h$ is the upper bound of all bids. 
	We have the optimal selling price $\opt \leq h$, implying that the approximation ratio is at most $h$. 
\end{proof}

We distinguish two cases to analyze the other bound: (1)~$\eta \leq \gamma$ and (2) $\eta>\gamma$. The main difference of the two cases is whether we can ensure that bidder $1$ is in $S$ or not. 

\begin{lemma}\label{lem:auction_small_eta}
	The approximation ratio of~\Cref{alg:auction} is at most~$\gamma \eta$ if $\eta \leq \gamma$.
\end{lemma}
\begin{proof}
	In this case, we can claim that bidder $1$ is always in $S$, i.e., $b_1 \geq br(1)$. If $br(1)=1$, $b_1$ is always at least $br(1)$. Otherwise, $br(1)=\predv_1 / \gamma$. Since $\optv_1 \geq \predv_1 /\eta$ and $\eta\leq \gamma$, we still have $b_1 \geq br(1)$.
	
	Suppose that the winner is bidder $j$ and the optimal payment is the private value $\optv_k$ of bidder $k$. 
	If $j=1$, the payment is $\theta(1) \geq br(1)$. Otherwise, the threshold bid of bidder $j$ is at least $b_1\geq br(1)$. Observing that regardless of which case, the payment is at least $br(1)$,
	we can bound the payment $\pay$ as follows:
	\[ \pay \geq br(1) \geq \frac{\predv_1}{\gamma} \geq \frac{\predv_k}{\gamma} \geq \frac{\optv_k}{\gamma\eta} = \frac{\opt}{\gamma\eta }, \]
	completing the proof.
\end{proof}

\begin{lemma}\label{lem:auction_large_eta}
	The approximation ratio of~\Cref{alg:auction} is at most $\max\{\gamma^2\eta^2,\frac{h\eta}{\gamma^2}\}$ if $\eta > \gamma$.
\end{lemma}

\begin{proof}
In this case, we cannot ensure that $b_1 \geq br(1)$. Notice that if bidder $1$ is still included in $S$, following the same analysis in~\Cref{lem:auction_small_eta} gives the ratio $\gamma\eta$. So we only need to focus on the case that $b_1 < br(1)$.
We further distinguish two subcases: (1) $\forall i \in [n]$, $\predv_i\geq \predv_1 / \gamma^2$, and (2) $\exists i \in [n]$, $\predv_i < \predv_1 / \gamma^2$. Basically, the two terms in the target approximation ratio come from these two subcases respectively. 

For the first subcase, a nice property is that for any two bidders $x,y\in [n]$,
	\[ \frac{\predv_x}{\predv_y} \geq \frac{\predv_x}{\predv_1} \geq \frac{1}{\gamma^2}. \] 
	Thus, we obtain the relationship between their real private values:
	\[\frac{\optv_x}{\optv_y} \geq \frac{1}{\eta^2} \frac{\predv_x}{\predv_y} \geq \frac{1}{\gamma^2 \eta^2}. \]
	Since the optimal value is one of the bidder's private values, for any bidder $x \in [n]$, 
	\[\optv_x \geq \frac{\opt}{\gamma^2 \eta^2}.\]
   Due to the assumption that $n\geq 3$, we have $|S| \geq 2$, implying that the threshold bid of the winner $j$ is at least some other bidder's bid. Thus, \[\pay = \theta(j) \geq \frac{\opt}{\gamma^2 \eta^2}.\]

For the second subcase, the payment might drop to $1$. 
But here, we can show that once the payment is $1$, the optimal payment will be at most $h\eta/\gamma^2$.
  We partition all bidders other than bidder $1$ into two groups: $A:=\{i>1 | \predv_i \geq \predv_1/\gamma^2 \}$ and $B:= \{ i>1 | \predv_i < \predv_1 /\gamma^2 \}$. If $A \cap S \neq \emptyset$, regardless of which group the winner is in, the threshold bid is at least $\predv_1/\gamma^2$. Namely, $\pay \geq \predv_1/\gamma^2$. Supposing that $\opt = \optv_k$, we have
   \[ \pay \geq  \frac{\predv_1}{\gamma^2}  \geq \frac{\predv_k}{\gamma^2} \geq \frac{\opt}{\gamma^2 \eta}.\]
   
   Finally, when $A \cap S = \emptyset$, we prove the claimed ratio by showing that $\opt \leq h\eta/\gamma^2$. 
   Since $A \cap S = \emptyset$, any bidder in $A$ has a bid below the bar $\predv_1/\gamma^2$. Namely, $\optv_i < \predv_1/\gamma^2$ $\forall i \in A$. Recalling that we only need to consider the case that $b_1 < br(1)$, for any bidder $i\in A\cup \{1\}$, $\optv_i < \predv_1 /\gamma.$
  
   Suppose that $\opt = \optv_k$.
   If $k\in A\cup \{1\}$, 
   \[ \opt =\optv_k < \frac{\predv_1}{\gamma} \leq \frac{h}{\gamma} < \frac{h\eta}{\gamma^2}.\]
   Otherwise, $k\in B$. Due to the definition of $B$,
   \[\opt =\optv_k \leq \eta \predv_k \leq \frac{\eta \predv_1}{\gamma^2}\leq \frac{h\eta}{\gamma^2}.\] 

\end{proof}

Combining the above lemmas, we can prove \Cref{thm:auction_main}.

\subsection{Lower Bounds}

This subsection gives two lower bounds of the auction. 

\begin{theorem}\label{thm:auction_jump_point}
	For any deterministic truthful mechanism with a consistency ratio $\gamma$, the approximation ratio is at least $h/\eta$ when $\eta > \gamma$.
\end{theorem}

Before proving \Cref{thm:auction_jump_point}, we first state a useful observation which is widely used in the lower bound proofs of auction mechanisms.

\begin{observation}\label{obs:lower_bound}
	Consider two bidders. For any deterministic DISC mechanism, given two bids $b_1,b_2$, if it lets bidder $1$ win and sets the payment to be $p$, the payment will still be $p$ if bidder $1$ increases $b_1$ unilaterally.
\end{observation}

\begin{proof}
	Assume for contradiction that $p$ changes when $b_1$ increases to $b_1'$.
	Due to \Cref{thm:myerson}, bidder $1$ will still win. If $p$ decreases, when $\optv_1=b_1$, bidder $1$ will misreport a bid $b_1'$ to reduce his payment. Otherwise, when $\optv_1 =b'_1$, bidder $1$ will also misreport a bid $b_1$ to pay less. Regardless of which case, there is a contradiction.
\end{proof}

\begin{proof}[Proof of \Cref{thm:auction_jump_point}]
	Consider an instance as follows. There are two bidders with the predicted values $\predv_1=h, \predv_2=1$. For any deterministic truthful mechanism $\cM$ with a consistency ratio $\gamma$, we claim that $\cM$ has to let bidder $2$ win when $\optv_1 < h/\gamma$ and $\optv_2=1$. Assume for contradiction that $\cM$ let bidder $1$ win. Since $\cM$ is truthful, $\pay \leq \optv_1 < h/\gamma$. Suppose that we increase $\optv_1$ to $h$. Then the prediction error becomes $1$ and the optimal value increases to $h$. However, due to \Cref{obs:lower_bound}, the payment should still be at most $\optv_1 < h/\gamma$. Hence, the approximation ratio is larger than $\gamma$, contradicting to the consistency of $\cM$.
	
	Since $\cM$ has to let bid $2$ win given  $\optv_1 < h/\gamma,$ $\optv_2=1$, the payment is at most $1$. Thus,
	\[\frac{\opt}{\pay} \geq \optv_1 = \frac{h}{\eta}.\]
\end{proof}

For any $\gamma \leq h^{1/5}$, when $\eta$ is slightly larger than $\gamma$, our ratio $f(\eta)$ is close to $h/\gamma$, which matches the lower bound given by \Cref{thm:auction_jump_point}.
To show the optimality of the mechanism's performance when $\eta \leq \gamma$, we assume that $\eta$ is sampled uniformly from $[1,\gamma]$ and prove a lower bound of the expected approximation ratio. 

\begin{theorem}\label{thm:auction_avg_ratio}
	For any deterministic truthful mechanism with a consistency ratio $\gamma$ $(< (\frac{2}{3}h)^{\frac{1}{4}})$, supposing that the prediction error $\eta$ is a random variable uniformly distributed on $[1,\gamma]$, the expected approximation ratio is at least $(\gamma+1)\gamma/2$, which is exactly the expected ratio of \Cref{alg:auction}.
\end{theorem}
\begin{proof}
    Without loss of generality, assume that $\gamma > 1$. 
	Consider an instance where there are two bidders with the prediction $\predv_1=h/\gamma, \predv_2=1$. Always set $\optv_2$ to be $1$.
	For any deterministic truthful mechanism $\cM$ with a consistency ratio~$\gamma$~$(<\sqrt{h})$, it has to let bidder $1$ win when $\optv_1=\predv_1=h/\gamma$. Thus, the threshold bid of bidder $1$ is at most $h/\gamma$. Define $\alpha:=(h/\gamma)/\theta(1)$. Since $\cM$ is $\gamma$-consistent, $\alpha\in [1,\gamma]$.
	
	We distinguish two cases according to $\alpha$: (1) $\eta \in [1,\alpha]$, and (2) $\eta \in (\alpha,\gamma]$. For the first case, set $\optv_1=\predv_1 \eta$. Since the threshold bid of bidder $1$ remains unchanged, 
	\[\frac{\opt}{\theta(1)} = \frac{\predv_1 \eta}{ \predv_1 / \alpha} = \eta\alpha.\]
	For the second case, set $\optv_1 = \predv_1 /\eta$. Due to the definition of $\alpha$, $\cM$ lets bidder $2$ win, implying that the approximation ratio is at least $h/(\gamma \eta) \geq h/\gamma^2$. 
	
	Now we are ready to analyze the expected approximation ratio. Supposing that $\cM$ has a ratio function $g(\eta)$, we have
	\begin{equation*}
		\begin{aligned}
		\E(g(\eta)) - \frac{\gamma(\gamma+1)}{2} &\geq \int_1^{\alpha} \frac{\alpha \eta}{\gamma-1}d\eta + \int_{\alpha}^{\gamma} \frac{h}{\gamma^2(\gamma-1)} d\eta - \frac{\gamma(\gamma+1)}{2} \\
		& = \frac{\alpha(\alpha^2-1)}{2(\gamma-1)} + \frac{h(\gamma-\alpha)}{ \gamma^2(\gamma-1)}  - \frac{\gamma(\gamma^2-1)}{2(\gamma-1)}\\
            & = \frac{1}{2(\gamma-1)} \cdot \left( \alpha^3-\alpha + \frac{2h(\gamma-\alpha)}{\gamma^2} - \gamma^3+\gamma \right)\\
            & = \frac{\gamma-\alpha}{2(\gamma-1)} \cdot \left( 1+ \frac{2h}{\gamma^2} - (\gamma^2+\gamma \alpha + \alpha^2) \right)\\
            & \geq \frac{\gamma-\alpha}{2(\gamma-1)} \cdot \left( 1+ \frac{2h}{\gamma^2} - 3\gamma^2 \right) \geq 0,
		\end{aligned}
	\end{equation*}
        where the last inequality used that $\gamma < (\frac{2}{3}h)^{\frac{1}{4}}$.
\end{proof}





\section{Frugal Path Auction}

In frugal path auction, there is a graph $G=(\{s,t\}\cup V,E)$ with at least two edge-disjoint $s$-$t$ paths. 
Each edge $e\in E$ is owned by an agent and the cost $\optc_e$ is a secret known only to that agent. 
Denote the number of agents by $n$.
Each agent $e$ submits a sealed bid~$b_e$. Then, the auctioneer selects an $s$-$t$ path $L$ as the winner and gives payment $p_e$ to each $e\in L$. If an agent wins, the utility is $p_e-\optc_e$. Otherwise, the utility is $0$. The goal is to design a truthful mechanism such that the total payment $\sum_{e\in E} p_e$ is minimized. In mechanism design with predictions, the auctioneer is given access to the edge cost predictions $\predC = \{\predc_e\}_{e\in E}$. 

For a frugal mechanism, a standard way to measure its performance is \emph{frugal ratio}. The formal definition of frugal ratio is a long statement (see~\cite{DBLP:conf/focs/KarlinKT05}). For simplicity, this section states an equivalent but more accessible definition specific to the path auction problem given in~\cite{DBLP:conf/soda/ArcherT02}. 
Say a frugal path mechanism~$\cM$. 
For an instance~$\cI$, use $P(\cI)$ to denote the total payment of $\cM$ and $S(\cI)$ to denote the cost of the second cheapest $s$-$t$ path.
The frugal ratio of a mechanism $\cM$ is defined to be $\fr = \sup_{\cI} P(\cI)/S(\cI)$.

For the brevity of the algorithm's statement and analysis, we assume that the graph consists of just some parallel $s$-$t$ paths and each edge has a positive cost. Notice that the assumptions do not affect the problem's hardness. The worst-case bound is still $\Omega(n)$ under the assumptions.  

\begin{theorem}\label{thm:path_auction_main}
	There exists a deterministic truthful mechanism parameterized by $\gamma \in [1,n^{1/3}]$ with frugal ratio at most $f(\eta)$,
	where the prediction error $\eta:=\max_{e\in E}\{\frac{\optc_{e}}{\predc_{e}}, \frac{\predc_{e}}{\optc_{e}}\}$ and 
	$$
	f(\eta):= \left\{
	\begin{aligned}
	&\gamma(1+\eta) & \eta \leq \gamma&  \\
	&\frac{n^2}{\gamma} & \eta > \gamma&. \\
	\end{aligned}
	\right.
	$$ 
\end{theorem}

\begin{algorithm}[tb]
	\caption{\;Frugal Path Auction with Predictions}
	\label{alg:path_auction}
	\textbf{Input}: A graph $G$, the predicted edge costs $\predC=\{\predc_e\}_{e\in E}$, the bids $B=\{b_e\}_{e\in E}$ and a parameter $\gamma \in [1,n^{1/3}]$. \\
	\textbf{Output}: The winner and the payments.  
	
	\begin{algorithmic}[1] 
		\STATE Say path $\predL$ is the path with the minimum predicted cost.
		\STATE For any $s$-$t$ path $L' \neq \predL$, set its weight $w(L')\leftarrow 1$.
		\IF { $\forall e \in \predL, b_e \leq \gamma \predc_e$ }
		\STATE Set $w(\predL)\leftarrow \gamma / n$.
		\ELSE
		\STATE Set $w(\predL)\leftarrow n / \gamma$.
		\ENDIF
		\STATE \textbf{return} the path $L$ with the minimum weighted bid and the threshold bid $\theta(e)$ of each edge $e\in L$.
	\end{algorithmic}
\end{algorithm}

\subsection{Mechanism}

The basic intuition is the same as the single-item auction problem. That is, we set a bar for the path with the minimum predicted cost and treat it differently when the reported cost is below or above the bar. Recollect that in single-item auction, if $b_1$ is below the bar, the mechanism ignores bidder 1 directly. 
However, things are different in path auction. 

Ignoring a path may lead to an unbounded frugal ratio. For example, consider a graph with three edge-disjoint $s$-$t$ paths. Their costs are $1$,$1$ and $\infty$ respectively. Say we ignore the first path. Then any truthful mechanism has to pay $\infty$, but the second cheapest path cost is only $1$. Hence, the frugal path mechanism needs a more careful design after setting up a bar. 
When the bid of the predicted cheapest path is far larger than the predicted cost,
instead of ignoring it directly, we assign a large weight to it so that the path is still a candidate but less competitive. The mechanism is described in \Cref{alg:path_auction}.

\subsection{Analysis}

The analysis starts by showing the truthfulness.

\begin{lemma}\label{lem:path_auction_DSIC}
	\Cref{alg:path_auction} is a truthful mechanism.
\end{lemma}

\begin{proof}
	The guarantee of the monotonicity is given by the weight function. For any edge, the weight of its path will not increase as its bid decreases. Thus, decreasing its bid only makes it more competitive, implying that the allocation rule is monotone. Due to \Cref{thm:myerson}, this lemma can be proved.
\end{proof}

For the ratio analysis, similar with the previous section, we distinguish two cases: (1) $\eta \leq \gamma$ and (2) $\eta > \gamma$. The intuition is that given a small $\eta$, we can ensure that $\predL$ always has a small weight.

\begin{lemma}\label{lem:path_auction_small_eta}
	The frugal ratio of \Cref{alg:path_auction} is at most $\gamma(1+\eta)$ if $\eta \leq \gamma$.
\end{lemma}
\begin{proof}
	Use $L'$ to denote the minimum cost path other than $\predL$. Clearly, path $L'$ has a cost at most the second cheapest path's cost $S$. Namely, $c(L') \leq S$. Observe that if $\eta \leq \gamma$, for each edge $e\in \predL$, \[b_e = \optc_e \leq \eta \predc_e \leq \gamma \predc_e.\] Hence, $w(\predL)=\gamma / n$.
	We show that $\predL$ will always win.
	
	\begin{equation*}
		\begin{aligned}
		w(\predL) b(\predL) &= \frac{\gamma}{n}\optc(\predL) \leq \frac{\gamma\eta}{n}\predc(\predL) \\
		& \leq \frac{\gamma\eta}{n}\predc(L') \;\;\;\; \mbox{[the definition of $\predL$]}\\
		& \leq \frac{\gamma\eta^2}{n}\optc(L') \\
		 & \leq w(L')b(L') \;\;\;\; \mbox{[$\eta \leq \gamma$ and $\gamma \leq n^{1/3}$].}\\
		\end{aligned}
	\end{equation*}
	
	Since $L'$ is the minimum cost path other than $\predL$, the weighted bid of $\predL$ is the minimum. Now we analyze the threshold bid of each edge in $\predL$. For each edge $e\in \predL$, if it loses when its bid is larger than $\gamma \predc_e$, its threshold bid $\theta(e) \leq \gamma \predc_e$. Otherwise, $\theta(e) > \gamma \predc_e$. When this edge increases the bid to $\theta(e)$, the weight of path $\predL$ becomes $n/\gamma$. Thus, 
	\[w(\predL)b(\predL) = \frac{n}{\gamma}(\optc(\predL)-\optc_e+\theta(e)). \]
	Since $\predL$ wins iff its weighted bid is at most $\optc(L')$, we have
	\[ \theta(e) = \frac{\gamma}{n} \optc(L') - \optc(\predL) + \optc_e. \]
	
	Use $Q_1\subseteq \predL$ and $Q_2\subseteq \predL$ to denote the edges in the above two cases respectively. We are ready to bound the frugal ratio.
	\begin{equation*}
		\begin{aligned}
			\fr &= \frac{\sum_{e\in E} p_e}{S} \leq \frac{\sum_{e\in \predL} \theta(e)}{\optc(L')}\\
			&\leq \frac{\sum_{e\in Q_1} \gamma \predc_e + \sum_{e\in Q_2} (\frac{\gamma}{n} \optc(L') - \optc(\predL) + \optc_e) }{\optc(L')} \\
			& \leq \frac{\gamma \predc(\predL) + \gamma \optc(L')}{\optc(L')} \leq \frac{\gamma \predc(L') + \gamma \optc(L')}{\optc(L')}\\
			& \leq \gamma(1+\eta).
		\end{aligned}
	\end{equation*}
	
\end{proof}

\begin{lemma}\label{lem:path_auction_large_eta}
	The frugal ratio of \Cref{alg:path_auction} is at most $n^2/\gamma$ if $\eta > \gamma$.
\end{lemma}

\begin{proof}
	Say path $L$, $L'$ and $L''$ are, respectively, the minimum weighted bid path, the second minimum weighted bid path and the second minimum cost path. 
	Recall that the frugal ratio is $(\sum_{e\in E}p_e)/c(L'')$. We distinguish two cases based on the weight of path $\predL$.

	In the first case that $w(\predL) =\gamma/n$, \[w(L') \optc(L') \leq \optc(L'').\]
	For any edge $e\in L$, we can always bound its threshold bid: \[\theta(e) \leq (n/\gamma) w(L')\optc(L').\] Thus,
	\[ \sum_{e\in E}p_e = \sum_{e\in L} \theta(e) \leq \frac{n^2}{\gamma}w(L')\optc(L') \leq \frac{n^2}{\gamma}\optc(L''). \]
	
	For the second case that $w(\predL) =n/\gamma$, since it is possible that $\predL=L''$, we have \[w(L')\optc(L') \leq (n/\gamma)\optc(L'').\] Regardless of which path wins, for any edge $e\in L$, the threshold bid $\theta(e) $ is at most $w(L')\optc(L')$. Thus,
	\[ \sum_{e\in E}p_e = \sum_{e\in L} \theta(e) \leq nw(L')\optc(L') \leq \frac{n^2}{\gamma}\optc(L'').  \]

\end{proof}

Combing the above lemmas, \Cref{thm:path_auction_main} has been proved.


Essentially, \Cref{alg:path_auction} can be viewed as a variant of the VCG path mechanism which incorporates predictions. Besides the VCG path mechanism, there is another famous mechanism for path auction called $\sqmechanism$-mechanism~\cite{DBLP:conf/focs/KarlinKT05}. Leveraging the idea of $\sqmechanism$-mechanism can improve the robustness slightly. In some graphs, the robustness ratio decreases by a factor of $\sqrt{n}$. See more details in \Cref{sec:path_auction_sqrt}.

\subsection{A Lower Bound}
We give a lower bound to show the optimality of \Cref{alg:path_auction} when $\eta$ is small.

\begin{theorem}\label{thm:path_auction_bound}
	For any deterministic truthful mechanism, the frugal ratio is $\Omega(\eta^2)$ given any $\eta=O(\sqrt{n})$.
\end{theorem} 

\begin{proof}
	Consider any deterministic truthful mechanism $\cM$ and any $\eta\geq 1$. Construct a graph $G=(V,E)$ consisting of only two node-disjoint $s$-$t$ paths $L_1,L_2$. The two paths share the same length. Thus, $|V|=|E|=n$ and $|L_1|=|L_2|=n/2$. The predicted cost of each edge is $1$.
	
	For each edge $e\in L_1$ (resp. $L_2$), say another edge $e'\in L_2$ (resp. $L_1$) is its competitor if $\cM$ will let $e'$ win when $b_e=b_{e'}=\eta$ and $b_{e''}=1/\eta$ $\forall e\in E\setminus \{e,e'\}$. Define edge set $T(e)$ to be the competitors of edge $e$. Noticing that for any pair $(e,e')\in L_1\times L_2$, either $e\in T(e')$ or $e'\in T(e)$, we have \[\sum_{e\in E} |T(e)| = |L_1|\cdot |L_2| = \frac{n^2}{4}.\] Thus, we can always find an edge $e$ with $|L(e)|\geq n/4$. Without loss of generality, assume $e\in L_1$.
	
	Set $\optc_e=\eta$ and $\optc_{e'} = 1/\eta$ for any other edge. 
	The mechanism has to let $L_2$ win due to the monotonicity. For each edge $e'\in T(e)\subseteq L_2$, the threshold bid is at least $\eta$ because it can still win when its bid increases to $\eta$. Thus, the payment is at least $|T(e)|\eta \geq n\eta /4$. The frugal ratio can be bounded as follows:
	\[\fr = \frac{\sum_{e\in E}p_e}{\optc(L_1)} \geq \frac{\frac{1}{4} n\eta }{\frac{1}{\eta}(\frac{n}{2}-1) + \eta} = \frac{n}{2n-1+\eta^2} \eta^2.\]
	Clearly, if $\eta=O(\sqrt{n})$, the frugal ratio is $\Omega(\eta^2)$.
\end{proof}

 When $\eta=\gamma-\epsilon$ for a tiny positive $\epsilon$, the frugal ratio of \Cref{alg:path_auction} approaches the lower bound. 


\section{Truthful Job Scheduling}\label{sec:scheduling}

\begin{algorithm}[tb]
	\caption{\;Truthful Job Scheduling with Predictions}
	\label{alg:scheduling}
	\textbf{Input}: Machines $M$, jobs $J$, the bid matrix $\{b_{ij}\}_{(i,j)\in M\times J}$, the predicted matrix $\{\predt_{ij}\}_{(i,j)\in M\times J}$ and a parameter $\gamma \in [1,m]$. \\
	\textbf{Output}: The allocation matrix $X$ and the payment matrix $P$.  
	
	\begin{algorithmic}[1] 
	    \STATE Construct a new processing time matrix $\barT$ as follows.
	    \FOR{each job $j\in J$}
	    \STATE Define $\predmint_j := \min_{i\in M} \predt_{ij}$.
	    \STATE For each machine $i\in M$,\\ \;\;\; set $\bart_{ij} \leftarrow \min \{ \predt_{ij},      (m/\gamma) \cdot \predmint_j \}$.
	    \ENDFOR
	    \STATE Compute the optimal allocation $\barX$ of the scheduling instance $\barT$.
	    \FOR{each job $j\in J$}
	    \STATE Use $k$ to denote the machine with $\barx_{kj}=1$.
	    \IF{ $\predt_{kj} > \bart_{kj}$}
	    \STATE Call job $j$ a greedy job.
	    \STATE Define $l:= \arg \min_{i\in M} \predt_{ij}$ and set up a weight function: $w_{lj} \leftarrow \gamma / m$ and $w_{ij} \leftarrow 1$ $\forall i \neq l$.
	    \ELSE
	    \STATE Call job $j$ a non-greedy job.
	    \STATE Set up a weight function: $w_{kj} \leftarrow \gamma^2 / m^2$ and $w_{ij} \leftarrow 1$ $\forall i \neq k$.
	    \ENDIF
	    \STATE Let the machine $z$ with the minimum weighted bid win and pay the threshold bid $\theta_{zj}$ for job $j$. Namely, set $x_{zj}\leftarrow 1, p_{zj} \leftarrow \theta_{zj}$ and $x_{ij}, p_{ij} \leftarrow 0$ $\forall i \neq z$.
	    \ENDFOR
		\STATE \textbf{return} $X=\{x_{ij}\}_{(i,j)\in M\times J}$ and $P=\{p_{ij}\}_{(i,j)\in M \times J}$.  
	\end{algorithmic}
\end{algorithm}

In this section, we consider the truthful job scheduling on unrelated machines. In the problem, there are $n$ jobs $J$ and $m$ machines $M$.  
Each machine $i$ is owned by a selfish agent, and the processing time $\optt_{ij}$ of each job $j$ on machine $i$ is a private value known only to the agent.
Without loss of generality, assume that $t_{ij}>0$. 
Each job has to be assigned to exactly one machine.
The mechanism will ask each machine $i$ to report a processing time $b_{ij}$ for each job $j$, and then allocate jobs. Once job $j$ is assigned to machine $i$, the mechanism pays the machine (agent) $p_{ij}$. Use $X=\{x_{ij}\}_{(i,j) \in M \times J}$ to denote the job allocation matrix, where $x_{ij}$ is $1$ if job $j$ is assigned to machine $i$, otherwise $0$. The makespan  of an allocation $X$ is the completion time of the last job, i.e., $\ms:= \max_{i\in M} \sum_{j\in J} x_{ij}\optt_{ij}$.
The utility of each machine~$i$ is defined to be $\sum_{j\in J} (p_{ij}-x_{ij}\optt_{ij})$. The goal is to design a truthful mechanism which minimizes the makespan. In mechanism design with predictions, we are given the predicted processing times $\predT=\{\predt_{ij}\}_{(i,j) \in M \times J}$. The prediction is error-prone and the prediction error $\eta:= \max_{(i,j)}\{\frac{\predt_{ij}}{\optt_{ij}}, \frac{\optt_{ij}}{\predt_{ij}} \}$.

This problem is very different from the previous two problems. First, it is not a single-parameter environment. For each machine $i$, the private value is a $n$-dimensional vector $<\optt_{i1},...,\optt_{in}>$. Second, the objective function is not directly related to the payment. 
The mechanism gives the winners money but the goal is to minimize the makespan. Thus, the payment rule in truthful scheduling is only for the purpose of ensuring the truthfulness.

\begin{theorem}\label{thm:scheduling_main}
	
	There exists a deterministic truthful mechanism parameterized by $\gamma\in [1,m]$ with approximation ratio $O(\min\{\gamma \eta^2, \frac{m^3}{\gamma^2}\})$.
\end{theorem}


\subsection{Mechanism}

We deal with each job independently. Thus, each job can be viewed as a single-item auction, and we can apply our technique without hurting the truthfulness. 
Similar with the previous two mechanisms, we retain the same job allocation when $\eta$ is small, in order to build the consistency and the error tolerance. We first compute the optimal assignment for the predicted instance, and set up a weight function of machines such that selecting the machine with the minimum weighted bid for each job gets the optimal assignment. Then we decrease the weights of some machines by a factor. This ensures that when $\eta$ is small, these machines still obtain the minimum weighted bids, and thereby the job assignment remains unchanged. For the robustness guarantee, we observe that if the ratio between any two machines' weights is at most $c$, the approximation ratio of the mechanism must be bounded by $c\cdot m$. Thus, we round the predicted instance such that after rounding, the ratio between any two machines' weights is bounded. 

Notice that using the rounded weights directly may hurt the consistency because the processing time of a job could be different in the predicted instance and the rounded instance. To maintain a good consistency and robustness, we only use the rounded weights to guide the jobs that share the same processing time in both of the two instances. For all other jobs, we set up their weights such that they can be assigned greedily when $\eta=1$. We describe the mechanism in~\Cref{alg:scheduling}.

\subsection{Analysis}



For simplicity, we assume that no two machines have the minimum weighted bid in each iteration. This assumption can be removed easily by giving priority to machine $k$ when breaking ties. We first show the truthfulness.

\begin{lemma}\label{lem:scheduling_DSIC}
\Cref{alg:scheduling} is a truthful mechanism.
\end{lemma}

\begin{proof}
Noticing that the mechanism assigns each job independently, we can view \cref{alg:scheduling} as $n$ independent single-item auctions, where $n$ is the number of jobs. In each auction, the auctioneer use the weighted VCG mechanism. Thus, machines report their real private values for each job, and the whole mechanism is truthful.
\end{proof}

We start the approximation ratio analysis by proving the consistency.

\begin{lemma}\label{lem:scheduling_consistency}
The approximation ratio of \Cref{alg:scheduling} is at most $1+(1-\frac{1}{m})\gamma$ if $\eta = 1$.
\end{lemma}

\begin{proof}
When $\eta=1$, the prediction is perfect, i.e., $\predt_{ij}=\optt_{ij}$ $\forall (i,j)\in M\times J$. Define $\ms(T)$ to be the minimum makespan of the scheduling instance $T$. 
Then, according to the construction of $\barT$, we have $\ms(\barT)\leq \ms(\predT) = \ms(\optT)$ in the case that $\eta=1$. 

Use $A,B$ to denote the non-greedy jobs and greedy jobs respectively.
We analyze the contributions of $A$ and $B$ to the makespan separately. We first show that \Cref{alg:scheduling} assigns each non-greedy job to the same machine as in $\barX$, which indicates that the total processing time of non-greedy jobs in any machine is at most $\ms(\barT)$.

For each non-greedy job $j\in A$, use $k$ to denote the machine with 
$\barx_{kj}=1$. 
Clearly, we have $\bart_{kj}=\predt_{kj} < (m/\gamma) \cdot \predmint_j$ and $w_{kj} = \gamma^2/m^2$. For any machine $i \neq k$,
\[w_{ij}b_{ij} = \predt_{ij} \geq \predmint_j > \frac{\gamma}{m} \predt_{kj} \geq w_{kj}b_{kj}. \]
Thus, \Cref{alg:scheduling} sets $x_{kj} = \barx_{kj} = 1$ and for any machine~$k$, 
\[\sum_{j\in A} \optt_{kj}x_{kj} = \sum_{j\in A} \bart_{ij}\barx_{kj} \leq \ms(\barT).\]

For each greedy job $j\in B$, use $l$ to denote the machine with the minimum $\predt_{lj}$. Since $w_{lj} = \gamma/m$ and $\eta=1$, \Cref{alg:scheduling} assigns job $j$ to machine $l$.
Suppose that job $j$ is assigned to machine $k$ in $\barX$.
We know that $\predt_{kj} > \bart_{kj} = (m/\gamma) \cdot \predt_{lj} $ due to $j\in B$. Since $\sum_{i\in M \setminus \{l\},j\in J} \bart_{kj} \barx_{kj} \leq (m-1)\ms(\barT)  $, we have for any machine $l$,
\[\sum_{j\in B} \optt_{lj}x_{lj} = \sum_{j\in B} \frac{\gamma}{m}\bart_{kj}\barx_{kj} \leq (1-\frac{1}{m})\gamma \cdot \ms(\barT).\]
Combining the above two inequalities completes the proof.
\end{proof}

\begin{lemma}\label{lem:scheduling_robust}
The approximation ratio of \Cref{alg:scheduling} is at most $m^3/\gamma^2$.
\end{lemma}

\begin{proof}
The critical point is the weight function.
Since $w_{ij} \in [\gamma^2/m^2,1]$ for each $(i,j)\in M\times J$, we have $w_{i_1j}/w_{i_2j} \leq m^2/\gamma^2$ for any $i_1,i_2 \in M$ and $j\in J$. Thus, for each machine $z \in M$, 
\[ \sum_{j\in J} \optt_{zj}x_{zj}  \leq \sum_{j\in J} \frac{m^2}{\gamma^2} \optmint_j x_{zj} \leq \frac{m^2}{\gamma^2} \sum_{j\in J} \optmint_j. \]

Noticing that $\frac{1}{m}\sum_{j\in J} \optmint_j$ is a lower bound of the minimum makespan, we have $ \sum_{j\in J} \optt_{zj}x_{zj} \leq \frac{m^3}{\gamma^2} \ms(\optT)$ for each machine $z\in M$.
\end{proof}

The last piece is proving the ratio $O(\gamma \eta^2)$. 

\begin{lemma}\label{lem:scheduling_eta}
The approximation ratio of \Cref{alg:scheduling} is at most $(1+2\gamma)\eta^2$.
\end{lemma}

\begin{proof}
Similar with the $\eta$ analysis in previous problems, we distinguish two cases according to the value of $\eta$: (1) $\eta \leq \sqrt{m/\gamma}$, and (2) $\eta > \sqrt{m/\gamma}$.

Use $\predX$ to denote the allocation matrix returned by \Cref{alg:scheduling} if all machines report their predicted values.
For case~(1), the mechanism returns the same allocation as $\predX$. To see this, we first observe that in $\predX$, if $\predx_{ij}=1$, $w_{ij}$ is either $\gamma/m$ or $\gamma^2/m^2$. Due to the construction of the weight function, it is easy to verify that for any of such $(i,j)$ pairs, \[\frac{m}{\gamma} w_{ij}\predt_{ij} \leq \min_{i'\neq i} w_{i'j}\predt_{i'j}.\] Thus, when $\eta \leq \sqrt{m/\gamma}$,
\begin{equation*}
    \begin{aligned}
    w_{ij}\optt_{ij} & \leq \eta \cdot w_{ij}\predt_{ij} \\
    &= \eta\cdot\frac{\gamma}{m}\cdot \min_{i'\neq i} w_{i'j} \predt_{i'j} \\
    & \leq \eta^2 \cdot \frac{\gamma}{m}\cdot \min_{i'\neq i} w_{i'j} \optt_{i'j} \\
    & \leq \min_{i'\neq i} w_{i'j} \optt_{i'j},
    \end{aligned}
\end{equation*}
implying that $x_{ij}$ is still $1$ in the current allocation.

Define $F(X, T)$ to be the makespan obtained by using allocation $X$ for the scheduling instance $T$.
An observation is that $F(X,T)$ increases by at most a factor of $\eta$ if each entry in $T$ is scaled up or down by at most a factor of $\eta$.
According to the observation and the consistency proof, we have
\[ F(X, \optT) \leq \eta F(X,\predT) \leq (1+\gamma)\eta \cdot \ms(\predT). \]
Use $\optX$ to denote the optimal assignment of $\optT$. Also due to the observation, 
\[\ms(\predT) \leq F(\optX,\predT) \leq \eta \cdot \ms(\optT). \]
Thus, \[F(X,\optT) \leq (1+\gamma)\eta^2 \cdot \ms(\optT).\]

For the second case, we partition all jobs into two sets according to whether they have the same assignments in $X$ and $\predX$: $J_1:=\{j\in J | x_{ij}=\predx_{ij} \; \forall i\in M\}$ and $J_2 := J \setminus J_1$. According to the analysis in the previous paragraph, the contribution of $J_1$ to the makespan is at most $(1+\gamma)\eta^2 \ms(\optT)$.

Finally, we consider the contribution of $J_2$. Recollect that in $\predX$, if $\predx_{ij}=1$, $w_{ij}$ is either $\gamma/m$ or $\gamma^2/m^2$. For any job $j\in J_2$, it must be assigned to a machine $z$ with $w_{zj}=1$ in $X$. Since machine $z$ has the minimum weighted bid,
$\optt_{zj} = w_{zj}b_{zj} \leq \optmint_j.$
Thus, for any machine $z\in M$,
\[ \sum_{j\in J_2} \optt_{zj}x_{zj} \leq \sum_{j\in J_2}\optmint_j \leq m \cdot \ms(\optT) \leq \gamma \eta^2 \cdot \ms(\optT). \]

\end{proof}

Combining the above lemmas proves \Cref{thm:scheduling_main}. 

\section{Two-Facility Game on a Line}

In this section, we consider the problem of locating two facilities on a real line to serve a set of selfish agents. In the problem, there are $n$ agents and each agent $i$ has a private location $\optv_i$. The mechanism asks each agent $i\in [n]$ to report the location $x_i \in \R$, and then outputs two facility locations $l_1,l_2 \in \R$. Each agent $i$ wants to minimize the connection cost $c_i$, which is the distance between $\optv_i$ and the nearest facility. The goal is to design a truthful mechanism which minimizes the total connection cost. In mechanism design with predictions, we are given the predicted location $\predv_i$ of each agent $i\in [n]$. 
We show that in an environment without money, we can still use predictions to improve the approximation ratio slightly.


\begin{theorem}\label{thm:facility}
There exists a truthful deterministic mechanism with $(1+n/2)$-consistent and $(2n-1)$-robust.
\end{theorem}


\begin{algorithm}[tb]
	\caption{\;Two-Facility Game with Predictions}
	\label{alg:facility}
	\textbf{Input}: The reported location profile $X=\{x_1,...,x_n\}$ and the predicted locations $\predV=\{\predv_1,...,\predv_n\}$. \\
	\textbf{Output}: The two facility locations.
	
	\begin{algorithmic}[1] 
	    \STATE Compute the optimal facility locations $\predl_1,\predl_2$ for $\predV$. We assume w.l.o.g.  that $\predl_1$ is the facility which attracts more agents and there exists an agent $\predi$ with $\predl_1=\predv_{\predi}$.
	    \STATE Set $l_1\leftarrow x_{\predi}$.
	    \STATE Define $d_A := \max \limits_{x_j \leq x_{\predi}} (x_{\predi}-x_j)$ and $d_B := \max \limits_{x_j \geq x_{\predi}} (x_j - x_{\predi})$.
	    \IF{$d_A \leq d_B$}
	    \STATE Set $l_2 \leftarrow l_1 + \max\{2d_A,d_B\}$.
	    \ELSE
	    \STATE Set $l_2 \leftarrow l_1 - \max\{d_A,2d_B\}$.
	    \ENDIF
		\STATE \textbf{return} $l_1$ and $l_2$.  
	\end{algorithmic}
\end{algorithm}

\subsection{Mechanism} 
The mechanism builds on the line mechanism~\cite{DBLP:conf/sigecom/LuSWZ10}. We notice that in the line mechanism, there exists an arbitrarily selected dictator.
Thus, a natural idea is leveraging predictions to select the dictator. 

An advantage of this idea is that the robustness can always be guaranteed because selecting any agent to be the dictator obtains an $O(n)$ approximation ratio. Then we check the structure of optimal solutions, and find that for any instance, there always exists an agent such that the line mechanism's performance can be improved if letting this agent be the dictator. Thus, setting the dictator to be such an agent in the predicted instance makes the mechanism consistent. The mechanism is described in \Cref{alg:facility}. 

\subsection{Analysis}






Since \Cref{alg:facility} is essentially the line mechanism with a dictator proposed in~\cite{DBLP:conf/sigecom/LuSWZ10}, the truthfulness and the robustness ratio are guaranteed from their analysis.

\begin{lemma}\label{lem:facility_DSIC_robust}
\Cref{alg:facility} is a truthful mechanism and has an approximation ratio at most $2n-1$.
\end{lemma}

Now we focus on the consistency proof.

\begin{lemma}\label{lem:facility_consistency}
The approximation ratio of \Cref{alg:facility} is at most $1+n/2$ if the predictions are error-free. 
\end{lemma}
\begin{proof}
Without loss of generality, we assume that $d_A \leq d_B$. 
Define $S_1$ and $S_2$, respectively, to be the agents connecting to $\predl_1$ and $\predl_2$ in the optimal solution. Since $\predl_1$ attracts more agents, $|S_1|\geq n/2$ and $|S_2 |\leq n/2$.
We distinguish two cases according to whether $\predl_1$ is the left facility: (1) $\predl_1 < \predl_2$ and $\predl_1>\predl_2$. 

Use $c_i$ and $\optc_i$ to denote the connection costs of agent $i$ in \cref{alg:facility} and the optimal solution respectively.
Note that regardless of which case, for any agent $i\in S_1$, \[c_i \leq |l_1-\optv_i| = |\predl_1 - \optv_i| = \optc_i\] when the predictions are error-free. Thus, we only need to analyze the connection cost of agent set $S_2$.

Say agent $a$ and $b$ are, respectively, the leftmost and the rightmost agents. In the first case, for any agent $i\in S_2$,
\begin{equation*}
    \begin{aligned}
    c_i &\leq l_2-\optv_i \\
    &\leq |\predl_2-\optv_i| + (l_2-\predl_2)\\
    &= \optc_i + l_2-\optv_b + \optv_b - \predl_2 \\
    &\leq \optc_i + d_A + \optv_b - \predl_2 \\
    \end{aligned}
\end{equation*}

Since $\predl_1 < \predl_2$, agent $a$ (resp. agent $b$) must be served by $\predl_1$ (resp. $\predl_2$) in the optimal solution. Thus, $\optc_a = d_A$ and $\optc_b = \optv_b -\predl_2$. We can get an upper bound of $c_i$: $c_i \leq \optc_i + \opt$.

In the second case, for any agent $i\in S_2$, 
\[c_i \leq |\optv_i-l_1| = |\optv_i-\predl_1| \leq d_B \leq \opt.\]
The last inequality holds because $\predl_1$ is the right facility, which definitely serves the rightmost agent.

Now we are ready to bound the total connection cost.

\begin{equation*}
    \begin{aligned}
    \sum_{i\in [n]}c_i &= \sum_{i\in S_1} c_i + \sum_{i\in S_2} c_i \\
    & \leq \sum_{i\in S_1} \optc_i + \sum_{i\in S_2} \optc_i + |S_2|\cdot \opt \\
    & \leq (1+\frac{n}{2})\opt. 
    \end{aligned}
\end{equation*}

\end{proof}
Combining the above lemmas completes the proof of \Cref{thm:facility}.

\section{Conclusion}

This paper initiates the systematic study of mechanism design with predictions, and develops a set of techniques that can be applied to a large number of classic mechanism design settings. There leave quite a lot of possibilities for future work. For example, on truthful job scheduling, whether such a trade-off between consistency and robustness is optimal (under Nisan-Ronen conjecture) is a very interesting open problem. 

For simplicity, we directly use private information as our predictions. One can of course study other maybe more limited information in the prediction model. The metric for prediction errors is another place one can introduce other possible modelings. 

We mainly focus on deterministic mechanisms in the paper. One can natural study how predictions help randomized mechanisms. Typically for maximization problems, one trivial mechanism is a random combination of the worst case mechanism and a mechanism which trusts the predictions completely. Can we do much better than this? 

To sum up, this is a wildly open research direction worth exploring.  





\newpage
\bibliographystyle{alpha}

\begin{thebibliography}{ADMW13}
	
	\bibitem[ADL09]{DBLP:conf/sigecom/AshlagiDL09}
	Itai Ashlagi, Shahar Dobzinski, and Ron Lavi.
	\newblock An optimal lower bound for anonymous scheduling mechanisms.
	\newblock In {\em {EC}}, pages 169--176. {ACM}, 2009.
	
	\bibitem[ADMW13]{DBLP:conf/soda/AzarDMW13}
	Pablo~Daniel Azar, Constantinos Daskalakis, Silvio Micali, and S.~Matthew
	Weinberg.
	\newblock Optimal and efficient parametric auctions.
	\newblock In {\em {SODA}}, pages 596--604. {SIAM}, 2013.
	
	\bibitem[AGKK20]{DBLP:conf/nips/AntoniadisGKK20}
	Antonios Antoniadis, Themis Gouleakis, Pieter Kleer, and Pavel Kolev.
	\newblock Secretary and online matching problems with machine learned advice.
	\newblock In {\em NeurIPS}, 2020.
	
	\bibitem[AGP20]{DBLP:conf/icml/AnandGP20}
	Keerti Anand, Rong Ge, and Debmalya Panigrahi.
	\newblock Customizing {ML} predictions for online algorithms.
	\newblock In {\em {ICML}}, volume 119 of {\em Proceedings of Machine Learning
		Research}, pages 303--313. {PMLR}, 2020.
	
	\bibitem[AT02]{DBLP:conf/soda/ArcherT02}
	Aaron Archer and {\'{E}}va Tardos.
	\newblock Frugal path mechanisms.
	\newblock In {\em {SODA}}, pages 991--999. {ACM/SIAM}, 2002.
	
	\bibitem[CEGP10]{DBLP:conf/focs/ChenEGP10}
	Ning Chen, Edith Elkind, Nick Gravin, and Fedor Petrov.
	\newblock Frugal mechanism design via spectral techniques.
	\newblock In {\em {FOCS}}, pages 755--764. {IEEE} Computer Society, 2010.
	
	\bibitem[CKK20]{DBLP:conf/stoc/0001KK20}
	George Christodoulou, Elias Koutsoupias, and Annam{\'{a}}ria Kov{\'{a}}cs.
	\newblock On the nisan-ronen conjecture for submodular valuations.
	\newblock In {\em {STOC}}, pages 1086--1096. {ACM}, 2020.
	
	\bibitem[CKV07]{DBLP:conf/soda/ChristodoulouKV07}
	George Christodoulou, Elias Koutsoupias, and Angelina Vidali.
	\newblock A lower bound for scheduling mechanisms.
	\newblock In {\em {SODA}}, pages 1163--1170. {SIAM}, 2007.
	
	\bibitem[DS20]{DBLP:journals/corr/abs-2007-04362}
	Shahar Dobzinski and Ariel Shaulker.
	\newblock Improved lower bounds for truthful scheduling.
	\newblock {\em CoRR}, abs/2007.04362, 2020.
	
	\bibitem[ESS04]{DBLP:conf/soda/ElkindSS04}
	Edith Elkind, Amit Sahai, and Kenneth Steiglitz.
	\newblock Frugality in path auctions.
	\newblock In {\em {SODA}}, pages 701--709. {SIAM}, 2004.
	
	\bibitem[FT14]{DBLP:journals/teco/FotakisT14}
	Dimitris Fotakis and Christos Tzamos.
	\newblock On the power of deterministic mechanisms for facility location games.
	\newblock {\em {ACM} Trans. Economics and Comput.}, 2(4):15:1--15:37, 2014.
	
	\bibitem[GHP20]{DBLP:conf/sagt/GiannakopoulosH20}
	Yiannis Giannakopoulos, Alexander Hammerl, and Diogo Po{\c{c}}as.
	\newblock A new lower bound for deterministic truthful scheduling.
	\newblock In {\em {SAGT}}, volume 12283 of {\em Lecture Notes in Computer
		Science}, pages 226--240. Springer, 2020.
	
	\bibitem[GHW01]{DBLP:conf/soda/GoldbergHW01}
	Andrew~V. Goldberg, Jason~D. Hartline, and Andrew Wright.
	\newblock Competitive auctions and digital goods.
	\newblock In {\em {SODA}}, pages 735--744. {ACM/SIAM}, 2001.
	
	\bibitem[GP19]{DBLP:conf/icml/GollapudiP19}
	Sreenivas Gollapudi and Debmalya Panigrahi.
	\newblock Online algorithms for rent-or-buy with expert advice.
	\newblock In {\em {ICML}}, volume~97 of {\em Proceedings of Machine Learning
		Research}, pages 2319--2327. {PMLR}, 2019.
	
	\bibitem[IKQP21]{DBLP:conf/spaa/Im0QP21}
	Sungjin Im, Ravi Kumar, Mahshid~Montazer Qaem, and Manish Purohit.
	\newblock Non-clairvoyant scheduling with predictions.
	\newblock In {\em {SPAA}}, pages 285--294. {ACM}, 2021.
	
	\bibitem[JPS20]{DBLP:conf/icalp/JiangP020}
	Zhihao Jiang, Debmalya Panigrahi, and Kevin Sun.
	\newblock Online algorithms for weighted paging with predictions.
	\newblock In {\em {ICALP}}, volume 168 of {\em LIPIcs}, pages 69:1--69:18.
	Schloss Dagstuhl - Leibniz-Zentrum f{\"{u}}r Informatik, 2020.
	
	\bibitem[KKT05]{DBLP:conf/focs/KarlinKT05}
	Anna~R. Karlin, David Kempe, and Tami Tamir.
	\newblock Beyond {VCG:} frugality of truthful mechanisms.
	\newblock In {\em {FOCS}}, pages 615--626. {IEEE} Computer Society, 2005.
	
	\bibitem[KV07]{DBLP:conf/mfcs/KoutsoupiasV07}
	Elias Koutsoupias and Angelina Vidali.
	\newblock A lower bound of 1+\emph{phi} for truthful scheduling mechanisms.
	\newblock In {\em {MFCS}}, volume 4708 of {\em Lecture Notes in Computer
		Science}, pages 454--464. Springer, 2007.
	
	\bibitem[LLMV20]{DBLP:conf/soda/LattanziLMV20}
	Silvio Lattanzi, Thomas Lavastida, Benjamin Moseley, and Sergei Vassilvitskii.
	\newblock Online scheduling via learned weights.
	\newblock In {\em {SODA}}, pages 1859--1877. {SIAM}, 2020.
	
	\bibitem[LSWZ10]{DBLP:conf/sigecom/LuSWZ10}
	Pinyan Lu, Xiaorui Sun, Yajun Wang, and Zeyuan~Allen Zhu.
	\newblock Asymptotically optimal strategy-proof mechanisms for two-facility
	games.
	\newblock In {\em {EC}}, pages 315--324. {ACM}, 2010.
	
	\bibitem[LV18]{DBLP:conf/icml/LykourisV18}
	Thodoris Lykouris and Sergei Vassilvitskii.
	\newblock Competitive caching with machine learned advice.
	\newblock In {\em {ICML}}, volume~80 of {\em Proceedings of Machine Learning
		Research}, pages 3302--3311. {PMLR}, 2018.
	
	\bibitem[LWZ09]{DBLP:conf/wine/LuWZ09}
	Pinyan Lu, Yajun Wang, and Yuan Zhou.
	\newblock Tighter bounds for facility games.
	\newblock In {\em {WINE}}, volume 5929 of {\em Lecture Notes in Computer
		Science}, pages 137--148. Springer, 2009.
	
	\bibitem[LX21]{DBLP:conf/icml/0001X21}
	Shi Li and Jiayi Xian.
	\newblock Online unrelated machine load balancing with predictions revisited.
	\newblock In {\em {ICML}}, volume 139 of {\em Proceedings of Machine Learning
		Research}, pages 6523--6532. {PMLR}, 2021.
	
	\bibitem[MV17]{DBLP:conf/nips/MedinaV17}
	Andres~Mu{\~{n}}oz Medina and Sergei Vassilvitskii.
	\newblock Revenue optimization with approximate bid predictions.
	\newblock In {\em {NIPS}}, pages 1858--1866, 2017.
	
	\bibitem[Mye81]{DBLP:journals/mor/Myerson81}
	Roger~B. Myerson.
	\newblock Optimal auction design.
	\newblock {\em Math. Oper. Res.}, 6(1):58--73, 1981.
	
	\bibitem[NR01]{DBLP:journals/geb/NisanR01}
	Noam Nisan and Amir Ronen.
	\newblock Algorithmic mechanism design.
	\newblock {\em Games Econ. Behav.}, 35(1-2):166--196, 2001.
	
	\bibitem[PSK18]{DBLP:conf/nips/PurohitSK18}
	Manish Purohit, Zoya Svitkina, and Ravi Kumar.
	\newblock Improving online algorithms via {ML} predictions.
	\newblock In {\em NeurIPS}, pages 9684--9693, 2018.
	
	\bibitem[PT09]{DBLP:conf/sigecom/ProcacciaT09}
	Ariel~D. Procaccia and Moshe Tennenholtz.
	\newblock Approximate mechanism design without money.
	\newblock In {\em {EC}}, pages 177--186. {ACM}, 2009.
	
	\bibitem[RI17]{DBLP:journals/eatcs/RoughgardenI17}
	Tim Roughgarden and Kazuo Iwama.
	\newblock Twenty lectures on algorithmic game theory.
	\newblock {\em Bull. {EATCS}}, 122, 2017.
	
	\bibitem[Roh20]{DBLP:conf/soda/Rohatgi20}
	Dhruv Rohatgi.
	\newblock Near-optimal bounds for online caching with machine learned advice.
	\newblock In {\em {SODA}}, pages 1834--1845. {SIAM}, 2020.
	
\end{thebibliography}

\newpage
\appendix
\section{The Two-Bidder Case}\label{sec:auction_n2}

We show the mechanism for the special case that $n=2$ in \Cref{alg:auction_n2}. Since increasing any bidder's bid can only improve the chance that he wins, we has the following lemma.

\begin{algorithm}[htbp]
	\caption{Single-Item Auction with Predictions $(n=2)$}
	\label{alg:auction_n2}
	\textbf{Input}: The predicted private values $\predV=\{\predv_1,\predv_2\}$, the bids $B=\{b_1,b_2\}$ and a parameter $\gamma \geq 1$. \\
	\textbf{Output}: The winner and the payment.  
	
	\begin{algorithmic}[1] 
		\STATE Reindex the bidders in the decreasing order of their predicted values.
		\IF { $\predv_2 < \predv_1 / \gamma^2$ }
		\STATE Set up a bar function: $br(1)\leftarrow \predv_1 /\gamma$ and $br(2) \leftarrow 1$.
		\STATE Define a bidder set $S:= \{ i \in [n] | b_i \geq br(i) \}$.
		\STATE Let the bidder $j\in S$ with the highest bid win.
		\ELSE
		\IF{$b_1 \geq \predv_1 / \gamma$}
		\STATE Let bidder $1$ win directly.
		\ELSE
		\STATE Assign a weight $w$ to each bidder: $w_1\leftarrow 1/\gamma^2$ and $w_2\leftarrow 1$.
		\STATE Define a bidder set $S:= \{ i \in [n] | b_i \geq 1 \}$.
		\STATE Let the bidder $i \in S$ with the highest $w_ib_i$ win.
		\ENDIF		
		\ENDIF
		\STATE \textbf{return} the winner $j$ and his threshold bid $\theta(j)$.
	\end{algorithmic}
\end{algorithm}

\begin{lemma}\label{lem:auction_DSIC_n2}
	\Cref{alg:auction_n2} is a truthful mechanism.
\end{lemma}

Similar with the proof of~\Cref{lem:auction_h}, we can show the approximation ratio is always at most $h$.

\begin{lemma}\label{lem:auction_h_n2}
	The approximation ratio of~\Cref{alg:auction_n2} is at most h.
\end{lemma}

Notice that when $\predv_2 < \predv_1 / \gamma^2$, the mechanism, as well as the ratio analysis, is exactly the same as~\Cref{alg:auction}. Thus, we only focus on the case that $\predv_2 \geq \predv_1 / \gamma^2$ in the following.

Analogously, we distinguish two cases: (1) $\eta \leq \gamma$ and (2) $\eta > \gamma$.

\begin{lemma}\label{lem:auction_small_eta_n2}
	The approximation ratio of \Cref{alg:auction_n2} is at most $\gamma \eta$ if $\eta \leq \gamma$.
\end{lemma} 
\begin{proof}
	As shown in the proof of \Cref{lem:auction_small_eta}, $b_1 \geq \predv_1 /\gamma$ when $\eta \leq \gamma$. Thus, bidder $1$ always wins. The tricky part is the threshold bid.
	We show that if bidder $1$ misreports a bid $x < \predv_1 / \gamma$, he cannot win.
	\begin{equation*}
		\begin{aligned}
		w_1x &< \frac{1}{\gamma^2} \cdot \frac{\predv_1}{\gamma} \\
		&=\frac{\predv_1}{\gamma^2} \cdot \frac{1}{\gamma}\\
		& \leq \frac{\predv_2}{\gamma } \;\;\;\; \mbox{[the condition that $\predv_2 \geq \predv_1 / \gamma^2$ ]} \\
		& \leq \frac{\eta \optv_2}{\gamma} \\
		& \leq w_2\optv_2 \;\;\;\; \mbox{[the condition that $\eta \leq \gamma$ ]}. \\		
		\end{aligned}
	\end{equation*}
	
	 Thus, we can bound the threshold bid of bidder $1$:
	\begin{equation*}
		\begin{aligned}
		\pay &= \theta(1) = \frac{\predv_1}{\gamma} \\
		& = \frac{\max\{\predv_1,\predv_2\}}{\gamma}\\
		& \geq \frac{\max\{\optv_1,\optv_2\}}{\gamma \eta} \\
		&= \frac{\opt}{\gamma \eta} .
		\end{aligned}
	\end{equation*}
\end{proof}

\begin{lemma}\label{lem:auction_large_eta_n2}
	The approximation ratio of \Cref{alg:auction_n2} is at most $\max\{\gamma^2\eta^2,\frac{h\eta}{\gamma^2}\}$ if $\eta > \gamma$.
\end{lemma}

\begin{proof}
	We further distinguish two subcases according to whether $b_1 \geq \predv_1 /\gamma$ or not. For the first subcase that $b_1 \geq \predv_1 /\gamma$, if bidder $1$ cannot win with a misreported bid $x < \predv_1 / \gamma$, the threshold bid is $\predv_1/\gamma$, leading to an approximation ratio of $\gamma \eta$. Otherwise, we see that the optimal payment must be $\optv_1$. Thus, 
	\begin{equation*}
		\begin{aligned}
		\theta(1) &= \frac{w_2}{w_1} \optv_2 \\
		 & = \gamma^2 \optv_2 \\
		& \geq \frac{\gamma^2}{\eta } \predv_2 \\
		&\geq \frac{\gamma^2}{\eta}\cdot \frac{\predv_1}{\gamma^2} \;\;\;\; \mbox{[the condition that $\predv_2 \geq \predv_1 / \gamma^2$ ]} \\
		& = \frac{\predv_1}{\eta} \\
		& \geq \frac{\optv_1}{\eta^2} \\
		& = \frac{\opt}{\eta^2} \\
		\end{aligned}
	\end{equation*}
	
	For the second subcase that $b_1 < \predv_1 /\gamma$, if bidder $1$ wins, we still have $\opt=\optv_1$ and $\theta(1)=w_2 \optv_2/w_1$, implying an approximation ratio of $\eta^2$. Now suppose that bidder $2$ wins. Clearly, the payment is $w_1\optv_1$. If $\opt=\optv_1$, the approximation ratio is at most $1/w_1=\gamma^2$. Otherwise, 
	\begin{equation*}
		\begin{aligned}
		\pay &= w_1\optv_1 = \frac{\optv_1}{\gamma^2} \\
		& \geq \frac{\predv_1}{ \gamma^2 \eta} \\
		& \geq \frac{\predv_2}{\gamma^2 \eta} \\
		& \geq \frac{\optv_2}{\gamma^2 \eta^2} \\
		& = \frac{\opt}{\gamma^2 \eta^2}.
		\end{aligned}
	\end{equation*}
\end{proof}

Combining \Cref{lem:auction_DSIC_n2}, \Cref{lem:auction_h_n2}, \Cref{lem:auction_small_eta_n2} and \Cref{lem:auction_large_eta_n2}, we see that \cref{thm:auction_main} holds when $n=2$.

\newpage
\section{Another Mechanism for Path Auction}\label{sec:path_auction_sqrt}

\begin{algorithm}[tb]
	\caption{\;$\sqmechanism$-Mechanism with Predictions}
	\label{alg:path_auction_sqrt}
	\textbf{Input}: A graph $G$, the predicted edge costs $\predC=\{\predc_e\}_{e\in E}$, the bids $B=\{b_e\}_{e\in E}$ and a parameter $\gamma \in [1,n^{1/3}]$. \\
	\textbf{Output}: The winner and the payments.  
	
	\begin{algorithmic}[1] 
		\STATE Find two paths $\predL_1$ and $\predL_2$ minimizing $\predc(\predL_1) + \predc(\predL_2)$.
		\STATE Define path $\predL \in \{L_1,L_2\}$ to be the path with the minimum $\sqrt{|\predL|}\predc(\predL)$. 
		\STATE For any $s$-$t$ path $L' \neq \predL$, set its weight $w(L') \leftarrow 1$.
		\IF { $\forall e \in \predL, b_e \leq \gamma \predc_e$ }
		\STATE Set $w(\predL) \leftarrow \gamma / n$.
		\ELSE
		\STATE Set $w(\predL) \leftarrow n / \gamma$.
		\ENDIF
		\STATE Find two paths $L_1$ and $L_2$ minimizing $w(L_1)b(L_1)+w(L_2)b(L_2)$. 
		\STATE \textbf{return} the path $L \in \{L_1,L_2\}$ with the minimum $\sqrt{|L|}w(L)b(L)$ and the threshold bid of each edge $e\in L$.
	\end{algorithmic}
\end{algorithm}

This section shows how to leverage the idea of $\sqmechanism$-mechanism in the learning-augmented path mechanism. The mechanism is stated in \Cref{alg:path_auction_sqrt}. With the similar analysis, we obtain the following theorem.
\begin{theorem}\label{thm:path_auction_sqrt}
	Let path $L_x$ and $L_y$ be the cheapest and the second cheapest paths respectively. \Cref{alg:path_auction_sqrt} is truthful and has a frugal ratio at most $f(\eta)$, where
	$$
	f(\eta)= \left\{
	\begin{aligned}
	&\gamma(1+\eta) & \eta \leq \gamma \\
	&\frac{n}{\gamma}\sqrt{|L_x|} \sqrt{ \max \{|L_y|,|\predL|\}} & \eta > \gamma \\
	\end{aligned}
	\right.
	$$.
\end{theorem}

In the worst case, \Cref{alg:path_auction_sqrt} shares the same robustness ratio $n^2/\gamma$ with \Cref{alg:path_auction}. But for some graphs (e.g. $|L_x|$ is small), \Cref{alg:path_auction_sqrt}  gives a better frugal ratio.

\end{document}